\documentclass[journal,comsoc]{IEEEtran}
\usepackage[T1]{fontenc}
\ifCLASSINFOpdf
\else
\fi
\usepackage{cite}
\usepackage{amsmath,amssymb,amsfonts}
\usepackage{graphicx}
\usepackage{textcomp}
\usepackage{xcolor}
\usepackage{epstopdf}
\usepackage{amsthm}
\usepackage{nicefrac}
\usepackage{balance}
\usepackage{tabularx,booktabs}
\usepackage{color}
\usepackage{subcaption}
\usepackage{mwe}
\usepackage{bbm}
\usepackage{algorithm} 
\usepackage{algpseudocode} 
\usepackage{xcolor}

\usepackage[cmintegrals]{newtxmath}
\newtheorem{theorem}{Theorem}

\newtheorem{lemma}[theorem]{Lemma}
\newtheorem*{definition}{Definition}
\newtheorem{remark}{Remark}
\newtheorem{property}{Property}

\begin{document}
\title{Latency and Reliability Trade-off with Computational Complexity Constraints: \\ OS Decoders and Generalizations}
\author{
	\IEEEauthorblockN{Hasan Basri Celebi$^\dagger$, Antonios Pitarokoilis$^\ddagger$, Mikael Skoglund$^\dagger$}
	\\
	\IEEEauthorblockA{
		$^\dagger$KTH Royal Institute of Technology, Stockholm, Sweden
		\\
		$^\ddagger$Ericsson AB, Stockholm, Sweden
	}
\thanks{This work was funded in part by the Swedish Foundation for Strategic Research (SSF) under grant agreement RIT15-0091.}
}

\maketitle

\begin{abstract}
	In this paper, we study the problem of latency and reliability trade-off in ultra-reliable low-latency communication (URLLC) in the presence of decoding complexity constraints. We consider linear block encoded codewords transmitted over a binary-input AWGN channel and decoded with order-statistic (OS) decoder. We first investigate the performance of OS decoders as a function of decoding complexity and propose an empirical model that accurately quantifies the corresponding trade-off. Next, a consistent way to compute the aggregate latency for complexity constrained receivers is presented, where the latency due to decoding is also included. It is shown that, with strict latency requirements, decoding latency cannot be neglected in complexity constrained receivers. Next, based on the proposed model, several optimization problems, relevant to the design of URLLC systems, are introduced and solved. It is shown that the decoding time has a drastic effect on the design of URLLC systems when constraints on decoding complexity are considered. Finally, it is also illustrated that the proposed model can closely describe the performance versus complexity trade-off for other candidate coding solutions for URLLC such as tail-biting convolutional codes, polar codes, and low-density parity-check codes.
\end{abstract}


\IEEEpeerreviewmaketitle

\section{Introduction}

Ultra-reliable low-latency communication (URLLC) is one of the three main service categories that have been defined in 5G, with the other two being enhanced mobile broadband and massive machine-type communication \cite{3gpp_new_radio}. URLLC provides communication support with stringent constraints on reliability and end-to-end latency and has attracted extensive attention and significant research interest, since information transmission with low-latency and high reliability is crucial for enabling various mission-critical services, such as machine-to-machine communication, remote surgery, augmented reality, vehicle automation, industrial robotics, factory automation, and smart-grid \cite{bennis_ultra_reliable}. 

Reliable communication is often characterized by channel capacity \cite{shannon_a_mathematical}. Since capacity is the ultimate error-free transmission rate as the codeword length becomes arbitrary large, it is mostly appropriate for latency-tolerant communication systems. In the existing literature, the performance of a latency-constrained communication system is often evaluated on the basis of the outage capacity \cite{yang_block_fading}, which is the maximal transmission rate  such that the probability of the instantaneous mutual information falling below this rate does not exceed a desired outage threshold. Similar to channel capacity, 
outage capacity is also most appropriate for arbitrarily large codewords\cite{verdu_a_general}. However, with the stringent latency constraints of URLLC systems, the assumption on arbitrarily large codeword blocklength cannot be justified \cite{ji_ultra_reliable}. Although research on maximal achievable transmission rates for finite blocklengths has a history going back to the 1960s \cite{gallager_a_simple}, a significant amount of progress has been achieved in the context of non-asymptotic information theory in the recent years (see \cite{molavianJazi_a_unified} and references therein). Non-asymptotic achievability and converse bounds for the finite blocklength regime are derived in \cite{polyanskiy_channel_coding}. It is shown that, compared to the asymptotic limits, a rate penalty needs to be paid when transmitting in the finite blocklength regime. This study attracted significant interest from the research community and several studies on the non-asymptotic achievable bounds for various channels with different fading environments have been published \cite{yang_quasi_static, huang_finite_blocklength, yang_beta_beta}.

Although the non-asymptotic achievable bounds reveal the theoretical limits, achieving them is still an open problem. Therefore, the selection of a channel encoding and decoding scheme that can perform close to the limit is significant in terms of increasing the transmission efficiency of the communication system. Several coding schemes that are suitable for URLLC are introduced in \cite{liva_code_design, celebi_wireless_comm, husain_channel_coding, sharma_polar_code, lian_performance, van_wonterghem_performance, jiang_latency_performance}. Their performances in the finite blocklength regime are also shown therein where performance of a decoder, in general, is identified according to its gap to the non-asymptotic limits. However, although it is observed that some channel coding schemes can perform very close to the limits, computational complexity is neither taken into account in the comparisons of the coding schemes nor in the derivation of the theoretical limits. 

There exists no generally accepted measure for the computational complexity of a typical channel decoder. Nevertheless, the total number of operations \textit{per-information-bit} is often selected as a metric for the computational complexity \cite{pfister_capacity_achieving, bae_capacity_achieving}. In \cite{sybis_channel_coding}, the computational complexity of several decoding algorithms, suitable for URLLC, is presented. Based on these results, it is shown in \cite{shivarnimoghaddam_short_block} that complexity of the coding schemes exponentially increases as they approach to the theoretical limits. It is also further shown that an excess power with respect to the theoretical limits must be spent to achieve a fixed allowed error rate at a fixed transmission rate, when a particular code is chosen. As discussed in \cite{kienle_on_complexity} and \cite{kestell_when_channel}, latency due to the computational complexity of a decoder is inversely proportional to the average computational power of a processor, in terms of speed. Therefore, a computationally intensive decoding process takes relatively longer duration in a complexity constrained receiver, such as low-budget IoT receiver. \cite{savage_complexity_of, savage_the_complexity}. In such applications, latency due to decoding is a significant determinant of decoder cost \cite{kim_ultrareliable}.

Latency due to the decoding of a packet is neglected in several studies as it is assumed that decoding happens instantaneously \cite{grover_fundamental_limits, hehn_ldpc_codes, rachinger_comparison_of, maiya_coding_with}. In \cite{hehn_ldpc_codes, rachinger_comparison_of} and \cite{maiya_coding_with}, even though the decoding latency is assumed to be negligible, the inevitable delays due to structural properties of low-density parity-check (LDPC) and convolutional codes are investigated and performance comparisons in case of equal structural delays are presented. A similar analysis on structural delay for learning-based coding schemes is recently presented in \cite{jiang_learn_codes}.  Decoding latency for the state-of-the-art codes such as LDPC and polar codes is investigated in \cite{wu_decoding_latency, qin_low_latency, fan_a_low}, in which low-complexity decoding schemes have been proposed. Recently, extended Bose, Ray-Chaudhuri, Hocquenghem (eBCH) codes \cite{reed_error_control} with order-statistic (OS) decoders \cite{fossorier_soft_decision} have gained interest of the research community due to their good performance in finite blocklength regime \cite{zaidi_5g_physical, yue_a_revisit}. It is shown that OS decoder performs close to the maximum likelihood (ML) decoder for linear block codes with substantially lower decoding complexity. To further reduce decoding latency, a low complexity decoding algorithm is proposed in \cite{choi_fast_and}.

This work differs substantially from the listed references as we consider the decoding latency as a performance metric for the system design for OS decoders. The goal of this paper is to investigate the maximal performance limits of short packet communications when the decoding complexity of OS decoders is taken into account. For this purpose, several significant design problems for URLLC applications are investigated. For instance, in order to decrease the aggregate latency, one may select an OS decoder with relatively lower complexity, which in turn may compromise the error probability of the decoder. This, therefore, reveals trade-offs among latency, computational complexity, and reliability. This implies that a refined modelling of these parameters must be considered. In this study, analyses on decoding latency and reliability are presented which are based on the \textit{per-information-bit} computational complexity of OS decoders.  Although the main focus is on OS decoders, we also discuss on the applicability of the proposed  model to the other families of codes. 

\textit{Contributions}: This work extends the authors' previous work on analyses of low-latency communication with computational complexity constrained OS decoders \cite{celebi_low_latency}. In this paper the following contributions are presented.
\begin{itemize}
	\item First, a consistent way to compute the aggregate latency due to the OS decoding process for complexity constrained receivers is presented.
	
	\item A mathematically tractable model that can accurately show the trade-off between the computational complexity of the OS decoder, in number of binary operations \textit{per-information-bit}, versus the excess power to the non-asymptotic achievability bound is introduced.
	
	\item With the help of this model, we address non-trivial optimization problems that are related to URLLC systems with OS decoders and computational complexity constraints. The following optimization problems are investigated:
	\begin{itemize}
		\item Given that a fixed number of information bits are intended to be transmitted under reliability and power constraints, what is the optimum selection of transmission parameters  that leads to the minimum aggregate latency?
		
		\item Given that a fixed number of information bits are intended to be transmitted under reliability, power, and latency constraints, what is the optimum selection of transmission parameters that leads to the minimum energy-per-bit?
		
		\item Under reliability, power, and latency constraints, what is the optimum selection of transmission parameters that leads to the maximum number of information bits to be transmitted?
	\end{itemize}
	
	\item It is also illustrated that other families of codes, such as tail-biting convolutional codes (TBCCs), LDPC, and polar codes, follow similar trends on the trade-off between computational complexity versus excess power, which implies that the proposed model can be adapted to be suitable for these families of codes as well.
\end{itemize} 

Solutions to the optimization problems reveal that the optimal parameter choices are directly associated with the constraints. Thus, the optimal design of a URLLC system is substantially influenced when decoding latency is taken into consideration. 

\textit{Notation}: Vectors and matrices are denoted by bold face lower and upper case letters, respectively. We use $ \mathcal{N}(\boldsymbol{\mu},\boldsymbol{\Sigma}) $ to denote independent real Gaussian random variables with mean $\boldsymbol{\mu}$ and covariance matrix $ \boldsymbol{\Sigma} $. All logarithms in this paper are with base 2 and $ \oplus $ represents the binary addition.

\section{System Model}\label{sec_system_model}

We consider communication over a discrete-time, binary-input AWGN (BI-AWGN) channel. A sequence of $ n $ symbols  $ \boldsymbol{x} = [ x_1, x_2, \ldots, x_n ], ~ x_i \in \{-1,+1\} $, which is termed as codeword, is transmitted over the channel. The observed sequence at the receiver is 
\begin{align}\label{eq_system_model}
\boldsymbol{y} = \sqrt{\rho}\boldsymbol{x} + \boldsymbol{z},
\end{align}
where $\boldsymbol{z}\sim\mathcal{N}(\boldsymbol{0},\boldsymbol{I}_n)$ and $\rho$ denotes the signal-to-noise ratio (SNR).

Without latency constraints, it is known that there exists a codebook, i.e., collection of codewords, with size $ 2^{nr}$ codewords, given that $ r < C $, where $ r $ is the transmission rate and $ C $ is called the channel capacity \cite{shannon_a_mathematical}, such that the codeword error probability (CEP)\footnote{That is the probability that the receiver decides in favor of a codeword that is different from the one actually sent.} vanishes as $ n \rightarrow \infty $. The capacity $ C $ of the channel in \eqref{eq_system_model} is given, as a function of $ \rho $, by \cite{erseghe_coding_in}
\begin{equation}\label{eq_channel_capacity}
C = \frac{1}{\sqrt{2\pi}} \int e^{-\frac{z^2}{2}}  \left( 1-\log\left( 1+e^{-2\rho+2z\sqrt{\rho}} \right)  \right) \mathrm{d}z.
\end{equation}
With strict latency constraints, however, i.e., when $ n $ is not allowed to take arbitrarily large values, the CEP is strictly positive and $ C $ overestimates the rate of reliable information transmission through a BI-AWGN channel. Recently, more refined upper and lower bounds on the maximal codebook size have been proposed for finite $ n $ and a non-zero CEP, $ \epsilon > 0 $ \cite{ polyanskiy_channel_coding}. Denote the maximal codebook size with codewords of length $ n $ and CEP $ \epsilon $ by $ 2^{n R^*} $, where $ R^* $ denotes the maximal transmission rate. Based on the bounds in \cite{polyanskiy_channel_coding}, it is shown that the maximal codebook size can be well approximated for a wide range of $ n $ and $ \epsilon $ by $ 2^{n R^*} \approx 2^{nR(n,\rho,\epsilon)} $, where
\begin{equation}\label{eq_normal_approximation}
R(n,\rho,\epsilon) = C - \sqrt{\frac{V}{n}} Q^{-1}(\epsilon) \log e + \mathcal{O} \left(\frac{\log n}{n}\right) .
\end{equation}
The quantity $ V $ is the channel dispersion, and for the channel in \eqref{eq_system_model} is
\begin{equation}\label{eq_channel_dispersion}
V = \frac{1}{\sqrt{2\pi}} \int e^{-\frac{z^2}{2}}  \Big( 1 - \log\left( 1+e^{-2\rho+2z\sqrt{\rho}} \right)  - C \Big)^2 \mathrm{d}z ,
\end{equation}
and $ Q^{-1}(\cdot) $ is the inverse of the Gaussian $ Q- $function
$ Q(x) = \int_{x}^{\infty}\frac{1}{\sqrt{2\pi}} e^{-\frac{t^2}{2}} \mathrm{d}t $ ,
and finally the big $ \mathcal{O}(\cdot) $ notation describes the limiting behavior of the third term as $ n \rightarrow \infty $. The expression in \eqref{eq_normal_approximation} is termed as the normal approximation to the maximal coding rate. In our analysis, we take the first two terms of \eqref{eq_normal_approximation} into account and treat it as if it were exact, with the implicit understanding that the terms of order $ \mathcal{O}(\cdot) $ and smaller are omitted.

From a communication point of view the aggregate latency of transmission of a codeword is the difference in time between the entry of a given bit to the communication interface at the transmitter and the time it exits the communication interface at the receiver. Therefore, it can be divided into three main parts:  \emph{i}) latency at the transmitter, \emph{ii}) transmission and propagation latencies, and \emph{iii}) latency at the receiver. Latency at the transmitter and receiver can be separated into many sub-parts such as latency due to the encoding, buffering, signal processing, interleaving and decoding. The transmission latency is the time required so that all the symbols of a codeword are sent into the channel, hence it is proportional to the codeword length. Latency due to signal propagation is inevitable due to physical constraints, it is a constant with respect to the choice of encoding/decoding scheme. Buffering, filtering, interleaving, etc. can also be  sources of latency, however, latency incurred by all these operations is not related to the choice of the encoding/decoding scheme. Hence, they simply add a small constant to the aggregate latency and therefore we neglect this constant in our further analysis. The main focus of this paper is to investigate the effect of decoding latency for complexity constrained receivers. Thus, we only focus on transmission and decoding latencies.

It is assumed that the transmission latency of a codeword is $ n T_s $ seconds, where $ T_s $ is the symbol duration. Thus, the aggregate latency $ L_A $ is considered as 
\begin{equation}\label{key}
L_A = n T_s + L_D ,
\end{equation}
with $ L_D $ denoting the decoding latency.\footnote{It is assumed that the decoding starts right after the whole codeword is received. A more indepth investigation is presented in \cite{hehn_ldpc_codes}.} It is one of the aims of this work to propose a model that describes in a general, accurate, and tractable way the latency introduced due to decoding based on the OS decoding algorithm given in \cite{fossorier_soft_decision}. 

\section{Modeling the Decoding Complexity}\label{sec_Decoding_Complexity_Modeling}

The decoding complexity model that is proposed in the present paper is based on linear block codes with OS decoders, originally presented in \cite{fossorier_soft_decision}. The three most important reasons of this selection can be listed as follows. \textit{i}) In prior works, \cite{liva_code_design, zaidi_5g_physical}, it has been shown that there exist linear block codes with OS decoders that can perform very close to the information-theoretic bounds for finite $ n $. \textit{ii}) The decoding performance of OS decoders can be easily parameterized by a single parameter, i.e., the order, $ s \in \mathbb{Q} $ \cite{fossorier_soft_decision}. \textit{iii}) Finally, the operations that are executed during decoding can be accurately tracked and the decoding complexity can be efficiently and intuitively described.

An uncoded binary information sequence $ \boldsymbol{u} = [u_1,u_2,\ldots,u_k], ~~ u_i \in \{0,1\} $, of $ k \leq n $ bits is mapped to an encoded binary sequence  $ \boldsymbol{b} = \boldsymbol{u} \boldsymbol{G} $, where $ \boldsymbol{G} \in \{0,1\}^{k \times n} $ is the generator matrix, of $ n $ bits which are then mapped to the transmitted codeword $ \boldsymbol{x} $ using the rule $ x_i = 2 b_i - 1 $. At the decoder, we consider the use of an OS decoder with order$ -s $. The components $ \{ y_i \}_{i=1}^n $ of the observed sequence $ \boldsymbol{y} $ are sorted in order of descending amplitudes and the hard-decoded $ k $ most-reliable bit sequence, $ \boldsymbol{r} $, is obtained. We denote the resulting permutation by $ \kappa(\cdot) $. The columns of $ \boldsymbol{G} $ are reordered by the same permutation, $ \kappa(\cdot) $, and Gauss-Jordan elimination is applied to form the corresponding systematic generator matrix $ \boldsymbol{G}_\kappa $.\footnote{It is possible that the first $k$ columns of the permuted $\boldsymbol{G}$ matrix can be linearly dependent. In this case, reaching to a new systematic generator matrix $\boldsymbol{G}_\kappa$ is not possible. Therefore, a second permutation is needed that will guarantee the first $k$ columns to be independent and $ |y'_1| \geq \cdots \geq |y'_k| $ and $ |y'_{k+1}| \geq \cdots \geq |y'_n| $, where $ y'_i $ represents the $ i $th element of the sorted $ \boldsymbol{y} $. Of course this may add some additional complexity terms. However, for the purpose of this paper, we neglect these additional terms.} Associated with $ s $, a list, $ \mathcal{L}_\text{TEP} $, of
\begin{equation}\label{eq_number_of_TEPs}
|\mathcal{L}_\text{TEP}| = \sum_{i=0}^{\lfloor s\rfloor} {{k}\choose{i}} + \left\lfloor \left(s - \lfloor s\rfloor\right)  {{k}\choose{\lfloor s\rfloor + 1} } \right\rfloor
\end{equation}
test error patterns (TEPs), i.e., bit sequences of length $ k $, denoted as $ \boldsymbol{e}_i $, is formed. This list includes all the TEPs with Hamming weight $ \leq \lfloor s \rfloor $ and the most probable TEPs with Hamming weight $ \lfloor s \rfloor + 1 $, which can be computed based on the probability of having $ \lfloor s \rfloor + 1 $ number of errors at different locations in the first $ k $ bits of the hard decoded $ \kappa(\boldsymbol{y}) $ \cite[Lemma 1]{yue_a_revisit}. The set of test codewords is then formed by mapping
\begin{equation}\label{key}
(\boldsymbol{r} \oplus \boldsymbol{e}_i) \boldsymbol{G}_\kappa,  ~~ \boldsymbol{e}_i \in \mathcal{L}_\text{TEP} .
\end{equation}
The test codeword that minimizes the Euclidean distance between the permuted sequence $ \kappa(\boldsymbol{y}) $ is selected as the most probable test codeword. The decoded information sequence is then produced by performing the inverse permutation, $ \kappa^{-1}(\cdot) $, and selecting the first $ k $ bits.

For notational consistency, given a fixed codebook $ \mathcal{C} $ containing $ 2^k $ codewords of length $ n $, we denote an OS decoder as $ \mathrm{d} ( \mathcal{C}, s, \rho) $, where it is meant that the OS decoder of order$ -s $ operates on the given codebook $ \mathcal{C} $ at SNR $ \rho $, and $ \epsilon(\mathcal{C}, s, \rho) $ denotes the achieved CEP with the codebook $ \mathcal{C} $ at SNR $ \rho $ with order$ -s $. Focusing on the computation intensive operations, the total number of binary operations \textit{per-information-bit} of an observed sequence, $ \boldsymbol{y}$, when the decoder $ \mathrm{d} ( \mathcal{C}, s, \rho) $ is used, can be calculated by \cite[Ch. 7.1]{gareth_linear_algebra}, \cite{baldi_on_the_use}\footnote{One can add additional complexity terms due to the sorting and inverse permutation processes. However, since their total complexities are relatively smaller compared to the terms in \eqref{eq_complexity_of_os_dec} for short blocklengths, we skip them and adopt \eqref{eq_complexity_of_os_dec} in further analysis.}
\begin{equation}\label{eq_complexity_of_os_dec}
K( \mathcal{C},s ) = nk + \frac{|\mathcal{L}_\text{TEP}|}{2} \left( n - q + \frac{qn}{k} \right) ,
\end{equation}
where $ q $ represents the number of quantization bits. The first term in \eqref{eq_complexity_of_os_dec} is due to the Gauss-Jordan elimination of the permuted $ \boldsymbol{G} $ matrix and the second term is due to the mapping of the set of test codewords and comparisons with $ \kappa(\boldsymbol{y}) $. 
When $ s < 2 $, \eqref{eq_complexity_of_os_dec} is dominated by the Gauss-Jordan elimination and for $ s \geq 2 $, the second part dominates the complexity. To address the limiting behavior of $ K( \mathcal{C},s ) $, one can use Stirling's approximation, given as $x! = \Gamma(x+1) \approx \sqrt{2\pi} x^{x+\frac{1}{2}} e^{-x}$, where $ \Gamma(\cdot) $ is the Gamma function, $ \Gamma(z) = \int_{0}^{\infty} x^{z-1} e^{x} \mathrm{d}x $. Implementing Stirling's approximation into binomial coefficient we have
\begin{equation}
\label{eq_binomial_stirling}
{{k}\choose{s} } \approx \left( 1-\frac{s}{k} \right)^{s-\frac{1}{2}} \frac{e^{-k}}{\left( 1-\frac{s}{k} \right)^k} \frac{k^s}{\Gamma(s+1)} = \mathcal{O}\left( {k^s} \right)
\end{equation}
since as $k \rightarrow \infty$ the first term, the denominator in the middle term, and $e^{-k}$ tend to 1. Then, we are only left with $\frac{k^s}{\Gamma(s+1)}$. Thus, the complexity order of OS decoder can be expressed as $ K( \mathcal{C},s ) = \mathcal{O}(nk^s) $.

The choice of the order, $ s $, limits the search space for the most probable test codeword by limiting the size of the list, $ \mathcal{L}_\text{TEP} $. In comparison to the ML decoder, that performs in general an exhaustive search over the codebook, which entails exponential complexity in $ k $, a choice of a moderate $ s $ leads to substantial reduction in decoding complexity. As a side comment, it is shown in \cite{fossorier_soft_decision} that the required order, $ s_r $, to achieve the ML decoder performance is approximately $ s_r = \min\left\lbrace \frac{d_{\min}}{4}-1, k \right\rbrace $, where $ d_{\min} $ denotes the minimum Hamming distance. 

Some easily verifiable properties hold for the relative performance of two decoders operating on the same codebook follows:
\begin{property}
	Let two decoders, $ \mathrm{d}_1(\mathcal{C},s_1, \rho) $ and $ \mathrm{d}_2(\mathcal{C},s_2, \rho) $, operate on the same codebook $ \mathcal{C} $ with $ s_1 \leq s_2 $ at the same SNR. It follows immediately by the selection of the TEP lists that $ \mathcal{L}_{\text{TEP},1} \subseteq \mathcal{L}_{\text{TEP},2} $, which implies $ K_1(\mathcal{C},s_1) \leq K_2(\mathcal{C},s_2) $ and $ \epsilon_1(\mathcal{C}, s_1, \rho) \geq \epsilon_2(\mathcal{C}, s_2, \rho) $ for every $ \rho $. Intuitively,  more complex decoder leads to lower CEP.
\end{property}
\begin{property}
	In addition, let two decoders $ \mathrm{d}_1(\mathcal{C},s, \rho_1) $ and $ \mathrm{d}_2(\mathcal{C},s, \rho_2) $, operate on the same codebook $ \mathcal{C} $ with different SNR levels given that $ \rho_1 \leq \rho_2 $. Then, it must be true that $ \epsilon_1(\mathcal{C}, s, \rho_1) \geq \epsilon_2(\mathcal{C}, s, \rho_2) $ where complexities of two decoders are the same. Intuitively, higher operating SNR leads to lower CEP.
\end{property}

Numerical performance results for OS decoders with orders $ s=\{0,1,2,3,4,5\} $ for $ n=128 $ and $ k=64 $ are shown in Fig. \ref{fig_CER_n_128_k_64}, where eBCH code with $ d_{\min}=22 $ is used for the encoding at the transmitter and the error bound is derived from \eqref{eq_normal_approximation}. Fig. \ref{fig_CER_n_128_k_64} shows that as the order$ -s $ increases, the performance of the decoder improves and it is near-optimal for $ s=s_r=5 $. However, as shown in (10) and (11), the increase in the order of the decoder leads to an exponential increase in the decoding complexity.

\begin{figure}[t]
	\centering
	\includegraphics[width=1\linewidth]{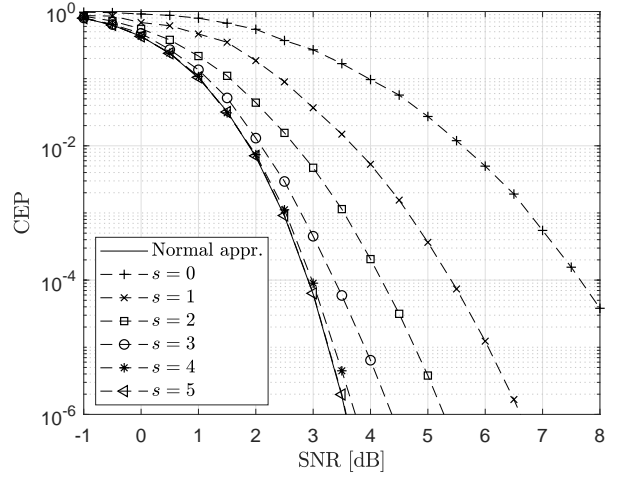}
	\caption{CEP performace of eBCH code with OS decoder at different orders compared to the the normal approximation error rate bound for BI-AWGN channel where $ n=128 $ and $ k=64 $.}
	\label{fig_CER_n_128_k_64}
\end{figure}

\section{Latency with Decoding Complexity Constraints}
\label{sec_complexity_constrained}

The transmission latency, $ L_T $, is proportional to the blocklength, $ n, $ however, the decoding latency depends on multiple parameters, i.e., $ (n,k,s), $ in a more complicated way. With stringent aggregate latency and reliability requirements, the optimal selection of the various parameters becomes a non-trivial task.

\subsection{Decoding Latency}

The decoding time of an OS decoder is influenced by a series of factors such as the particular hardware platform. For simplicity and generality, we assume that the binary operations are handled sequentially by the processor which leads to a linear relation between the total decoding duration and the time required for a binary operation on the hardware platform, denoted as $T_b$ \cite{bistritz_an_efficient}. Thus, the total transmission and decoding latency for the transmission of a codeword of blocklength $n$ can be written as
\begin{equation}
L_A = n T_s +  k K(\mathcal{C},s) T_b ,
\label{eq_total_latency_with_c}
\end{equation}
where $ k K(\mathcal{C},s) $ is the total number of binary computations required for decoding. A more accurate estimation on the decoding time can be done by investigating the algorithmic efficiency \cite{kienle_on_complexity} and specific hardware specifications such as, memory timings, buffer management, etc. \cite{wilhelm_the_worstcase}. Even though such an investigation is out of the scope of this paper, a discussion on the availability of parallel processing and its effect on \eqref{eq_total_latency_with_c} is included in Section \ref{sec_discussions}.

Suppose the latency constraint is
\begin{equation}\label{eq_latency_deadline}
L_A \leq L_M,
\end{equation}
where $ L_M $ represents the maximum allowed latency. In general, two different communication strategies can be considered in URLLC applications: \textit{i}) continuous-mode, \textit{ii}) sporadic/bursty-mode \cite{bennis_ultra_reliable, kim_ultrareliable}. A typical assumption for continuous-mode transmission is that decoding resources are chosen so that $ L_D $ is upper-bounded by $ nTs $. 
On the other hand, as discussed in \cite{kim_ultrareliable} and \cite{feng_dynamic_resource}, URLLC traffic can be event-driven, therefore, sporadic. The focus in this paper is on sporadic/bursty-mode communication, where a single packet is transmitted with a latency constraint on $ L_A $. Thus, depending on $ L_A $, $ L_D $ can be longer that $ nT_s $.  

The constraint in \eqref{eq_latency_deadline} imposes an upper bound on the \textit{per-information-bit} decoder complexity such that
\begin{equation}\label{eq_upper_bound_on_K}
K(\mathcal{C},s) \leq \frac{L_M-nT_s}{kT_b} ,
\end{equation}
as long as $ L_M \geq nT_s $. For fixed $ n $ and $ k $, the constraint in \eqref{eq_latency_deadline} restricts the order$ -s $ as follows
\begin{equation}
s \leq s_m = \underset{\{s| s \in \mathbb{Q^+}, ~ L_A \leq L_M \}}{\mathrm{arg~max}} K(\mathcal{C},s),
\label{eq_general_bound_on_s}
\end{equation}
where $ s_m $ denotes the maximum allowed order. Due to the sum of binomial coefficients, used while calculating $ |\mathcal{L}_{\text{TEP}}| $, a closed-form expression on $ s_m $ does not appear to be obtainable. However, an upper bound on the \textit{per-information-bit} complexity that gets tighter with larger $ s $ can be derived by using \cite[Lemma 3.6]{gray_entropy}
\begin{equation}\label{eq_upper_bound_on_K_2}
K(\mathcal{C},s) \leq nk + 2^{k h\left( \frac{s_m}{k} \right)-1} \left( n - q + \frac{qn}{k} \right),
\end{equation}
where  $ h(z) = -z \log(z) - (1-z) \log(1-z) $ is the binary entropy function. From \eqref{eq_upper_bound_on_K_2} we get
\begin{equation}\label{eq_bound_with_entropy_func}
h\left( \frac{s_m}{k} \right) \geq \frac{1}{k} \left( 1+\log \tau \right) , 
\end{equation}
where $ \tau = \frac{1}{n(k+q)-qk} \Big( \frac{L_M - nT_s}{T_b} - nk^2 \Big)  $. A lower bound on $ s_m $ can be numerically evaluated from \eqref{eq_bound_with_entropy_func} since the binary entropy function is monotonically increasing for $ \frac{s_m}{k} \leq \frac{1}{2} $. Finally, using the tight approximation for binary entropy function, $ h(z) \approx \left( 4z(1-z) \right)^{3/4} $, we obtain
\begin{equation}
s_m \approx \frac{k}{2} \left( 1 - \sqrt{1 - \left( \frac{1 + \log \tau}{k} \right)^{4/3}} \right).
\label{eq_bound_on_s}
\end{equation}

However, we note that a constraint on order$ -s $ may lead to a degradation in the CEP performance of the OS decoder. In particular, if $ s_m < s_r $, the CEP of the most complex allowable decoder will be appreciably higher in comparison to the ML CEP bound.

\subsection{Power Penalty} 

It is shown in \eqref{eq_general_bound_on_s} that the selection of an order$-s$ for a particular code of fixed $n$, $k$ and $\rho$ can be used to control the aggregate latency $L_A$ of the communication, albeit at the expense of reduced reliability. In Fig. \ref{fig_CER_n_128_k_64} for a fixed SNR the lowest CEP is given by the $ \epsilon_m $ curve. Constraining the order$ -s $ of the decoder, though, incurs a CEP degradation that is a vertical upwards step to the curve with corresponding $ s $. In order to satisfy a desired target reliability, a power penalty, i.e., an amount of excess power, has to be paid. Visually, this can be represented as a horizontal rightward step in Fig. \ref{fig_CER_n_128_k_64}. Hence, an interesting, yet complex, relation between power, aggregate latency, decoding complexity arises.

\begin{definition}[Power penalty]
	Fix a codebook $ \mathcal{C} $ of $ 2^k $ codewords of blocklength $ n $. For a reference SNR, $ \rho_r $, consider the CEP, $ \epsilon  $, given by the normal approximation in \eqref{eq_normal_approximation} and the suboptimal decoder $ \mathrm{d}(\mathcal{C},s, \rho)  $ that achieves $ \epsilon $ at SNR $ \rho $. The quantity  
	\begin{equation}
	\Delta\rho = \rho - \rho_r
	\end{equation}
	is the power penalty required, such that the  suboptimal decoder can achieve the same CEP as \eqref{eq_normal_approximation}.
\end{definition} 

For a fixed rate $ r=\frac{k}{n} $ the reference SNR $ \rho_r $ can be computed by taking the inverse of \eqref{eq_normal_approximation},
\begin{equation}\label{eq_reference_power}
\rho_r = R^{-1}(n,r,\epsilon) .
\end{equation}
Notice that $ R(n,\rho,\epsilon) $ is strictly increasing in $ \rho $ and therefore invertible. Although there is no closed form expression of $ \rho_r $ for BI-AWGN channels, it can be numerically evaluated. 

Since the ML decoder minimizes CEP, it holds that $ \Delta\rho \geq 0 $. For a family of codes that does not achieve the bound even with ML decoding, it holds $ \Delta\rho\geq \Delta\rho_{\text{ML}}>0 $, where $\Delta\rho_{\text{ML}}$ is the power gap of the best code within that family of codes from the normal approximation. Theoretically, it is possible to operate at rate $ r $ with SNR $ \rho_r $ (or $ \rho_r+\Delta\rho_{\text{ML}} $ for codes not achieving the bound), however, a possibly prohibitively complex decoder is required for such a power-rate selection. Similar empirical results are also presented in \cite{shivarnimoghaddam_short_block} and \cite{celebi_low_latency}. Extensive studies on OS decoders reveal that this exponential increase is similar at all rates for fixed $ n $ \cite{celebi_low_latency}. 

It is clear from the above that a model is required to quantify the power penalty. Bounds on the performance of OS decoders are available \cite{fossorier_soft_decision,dhakal_on_the}, however, these bounds are mathematically intractable for further analytical analysis. In Fig \ref{fig_complexity_powerpenalty_tradeoff}.a we plot the required SNR values so that an OS decoder of order $s$ achieves CEP $\epsilon(\mathcal{C}, s, \rho) = 10^{-5}$ for a codebook of blocklength $n=128$ and various rates. These values were computed via extensive simulations of the respective codes. For the purpose of comparison, we also show the capacity of the BI-AWGN channel with the dashed line and the normal approximation with the solid line. Fig. \ref{fig_complexity_powerpenalty_tradeoff}.a illustrates the following: \textit{i}) The performance of OS decoders closely approaches $ R(n,\rho,\epsilon) $ if $ s $ is sufficiently high. \textit{ii}) As the decoding complexity increases with increasing $ s $, the power penalty required for the desired CEP decreases. \textit{iii}) Conversely, an aggregate latency constraint, which implies a decoding complexity constraint, i.e., an upper bound on the order $ s $, leads to a corresponding power penalty, if a desired CEP is to be guaranteed.

\begin{figure}
	\centering
	\begin{subfigure}[b]{0.49\textwidth}
		\includegraphics[width=1\linewidth]{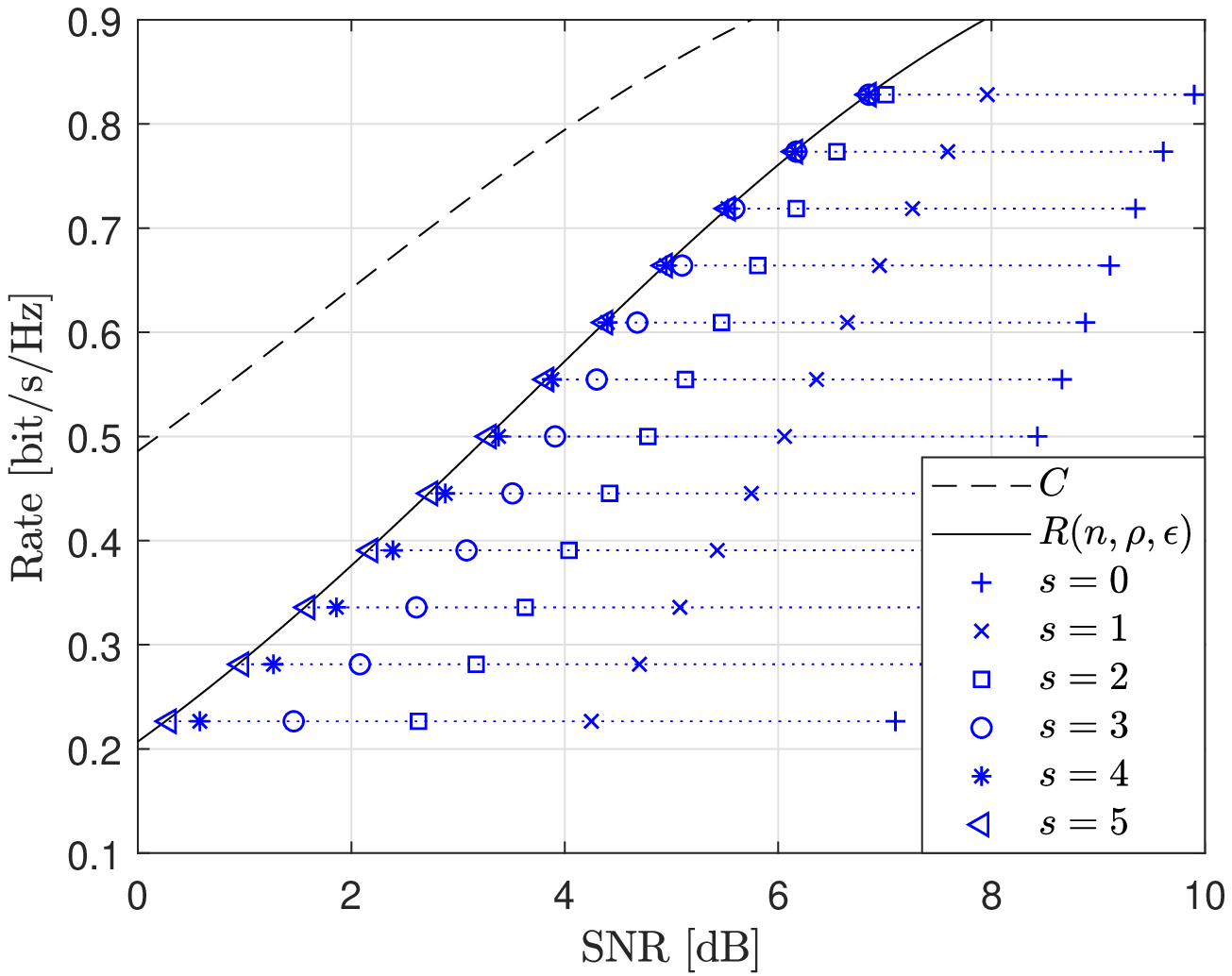}
		\caption{}
		\label{fig_power_vs_rate_w_decoders} 
	\end{subfigure}
	
	\begin{subfigure}[b]{0.49\textwidth}
		\includegraphics[width=1\linewidth]{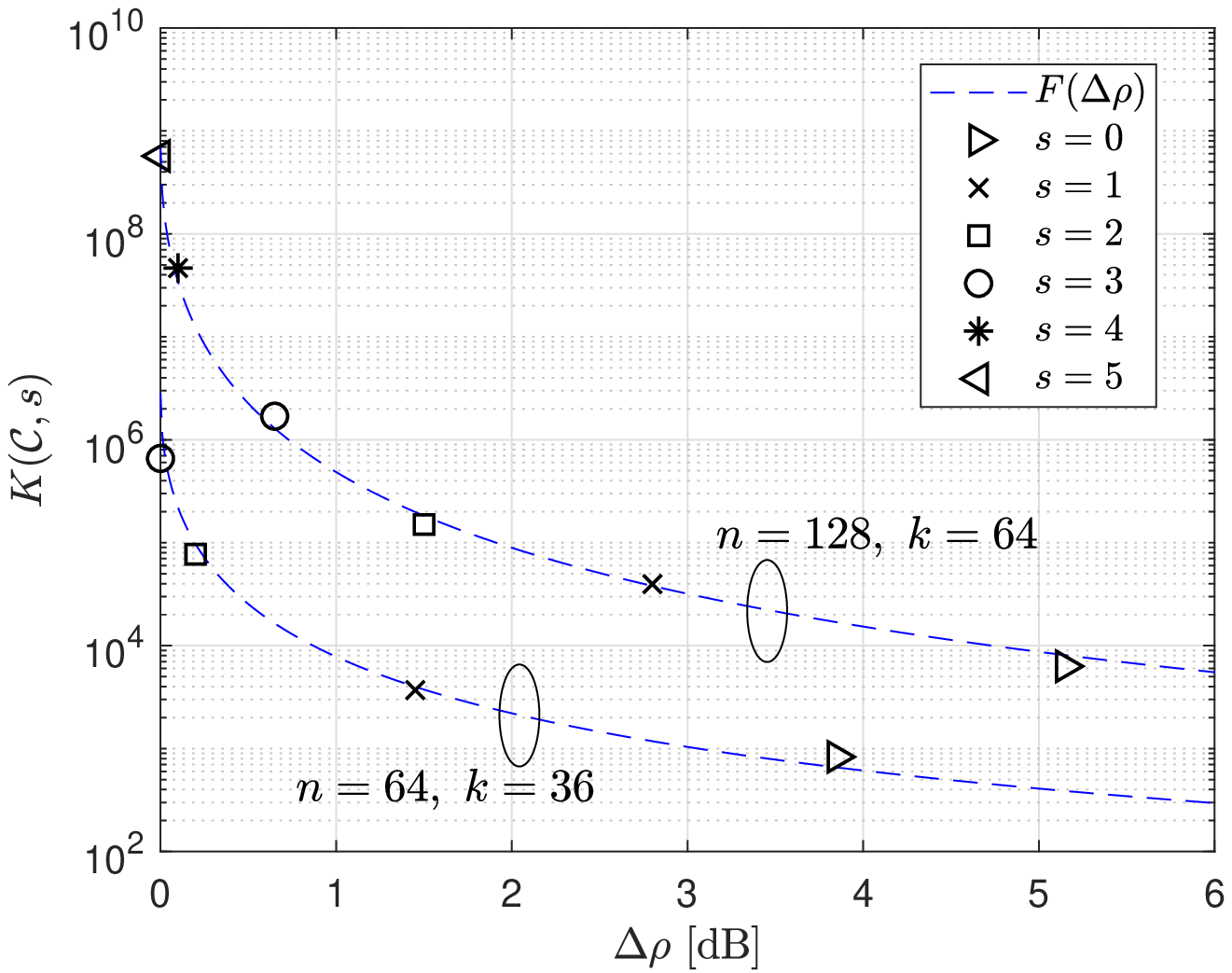}
		\caption{}
		\label{fig_q_vs_deltaP_n_128_64}
	\end{subfigure}
	
	\caption{(a). Power requirements of OS decoders with different
		orders at different rates for $ \epsilon=10^{-5} $ when $ n=128 $. (b). Power penalty values of OS decoders at different orders versus their complexities for $ n=128 $, $ k=64 $, and $ n=64 $, $ k=36 $ where $ \epsilon=10^{-5} $.}
	\label{fig_complexity_powerpenalty_tradeoff}
\end{figure}

In Fig. \ref{fig_complexity_powerpenalty_tradeoff}.b the total number of binary operations \textit{per-information-bit} is plotted as a function of the power penalty for order$ -s = \{0,1,2,3,4,5\} $ where $ n=128 $, $ k=64 $ and $ n=64 $, $ k=36 $ codes. Similar numerical results have also been produced for various $ n $ and $ k $ values for fixed $ \epsilon $ and it has been observed that for all cases, the relation between the logarithm of the computational complexity, $ \log K( \mathcal{C},s ) $, and power penalty can be modeled by a law of the type 
\begin{equation}
F(\Delta\rho) \overset{\Delta}{=} \frac{1}{a \sqrt{\Delta\rho} + b},
\label{eq_q_delta_tradeoff}
\end{equation}
with appropriate choices of the constants $a > 0 $ and $b > 0 $. This model describes in an accurate and tractable way the trade-off between decoding complexity and power penalty for practical finite-length codes. The coefficients $ a $ and $ b $ can be found with an iterative approach that searches the values which minimize the mean square error between the logarithmic computational complexity of decoders at the specified power gap and the model $ F(\Delta\rho) $. 

Based on extensive numerical simulations, it is observed that for fixed $ n $, the  values of $a$ and $b$ do not appreciably change as $ k $ varies. Therefore, for simplicity, we assume that $a$ and $b$ are functions of $ n $ only. As $ \Delta\rho \rightarrow 0 $, $ F(\Delta\rho) = 1/b $ and the ultimate complexity of an OS decoder that can achieve the benchmark is $ \approx 2^{1/b} $. Given that $ a $ and $ b  $ are strictly positive, $ F(\Delta\rho) $ is a monotonically decreasing function in $ \Delta\rho $, since 
\begin{equation}\label{eq_f_deltaP_mon_decrease}
F' = -\frac{a}{2 \sqrt{\Delta\rho} \left( a\sqrt{\Delta\rho} + b \right)^2} < 0 .
\end{equation}
The monotonicity of $ F(\Delta\rho) $, which follows from \eqref{eq_f_deltaP_mon_decrease} is not imposed by the authors, but is a consequence of the behaviour based on Fig. \ref{fig_complexity_powerpenalty_tradeoff}.b and is a direct consequence of the decoder's operation as given in Properties 1 and 2. Further, \eqref{eq_f_deltaP_mon_decrease} reveals that a desired CEP can be achieved with a lower complexity decoder as long as sufficient excess power is available, and vice-versa.

\begin{lemma} \label{lemma_min_power_penalty}
	Consider the system model described in Sec. \ref{sec_system_model} and let a constraint $L_A\leq L_M$ with $L_M > nT_s$  imposed on a complexity constrained OS decoder, where the aggregate latency is expressed as \eqref{eq_total_latency_with_c}. Based on the proposed model in \eqref{eq_q_delta_tradeoff}, the minimum amount of power penalty that is required to guarantee a desired CEP is 
	\begin{equation}\label{eq_min_req_power_penalty}
	\Delta\rho_m = \left( \frac{1}{a} \max \left\lbrace \left( \log \frac{L_M-nT_s}{kT_b} \right)^{-1} -b , 0\right\rbrace \right)^2 .
	\end{equation}
\end{lemma}

\begin{remark}
	Lemma \ref{lemma_min_power_penalty} shows the minimum amount of excess power that is needed in order to fulfill the latency and reliability requirements for a complexity constrained OS decoder with BI-AWGN channel.  From (\ref{eq_min_req_power_penalty}), it is clear that for fixed $ n $ and $ T_s $, as $ T_b $ decreases, i.e., the receiver is equipped with a more powerful processor, $ \Delta\rho_m $ decreases and hence the gap to the normal approximation shrinks and vanishes if $ T_b \leq \frac{L_M-nT_s}{k \sqrt[b]{2}} $. On the other hand, for fixed $ n $, if the transmission rate, $ r $, increases, $ \Delta\rho_m $ also increases and the gap to the normal approximation widens.
\end{remark}

The latest argument expressed in Remark 1 can be explained as follows. Recall that selecting the maximum allowed $ K(\mathcal{C},\rho,s) $ leads to the minimum amount of power penalty and, based on the upper bound on \textit{per-information-bit} decoder complexity given in \eqref{eq_upper_bound_on_K}, for fixed $ n $, as $ k $ increases, i.e., when transmitting at higher rates, \eqref{eq_upper_bound_on_K} decreases. Thus, in order to assure this inequality, as $ k $ increases, a simpler decoder, with smaller $s$, is required,  which eventually leads to higher power penalty.

\subsection{Maximal Information Rate with Latency Constraints}

Here, an approximation on the maximal information rate that can be achievable under latency, reliability, and complexity constraints is presented.

\begin{lemma} \label{lemma_max_inf_rate}
	For a complexity constrained receiver with aggregate latency expressed in \eqref{eq_total_latency_with_c}, the maximal achievable information rate subject to latency, $L_A<L_M$ with $L_M>nT_s$, and reliability constraints, denoted as $ M^* $, can be closely approximated as $ M^* \approx M(n,\rho, \epsilon) $ where 
	\begin{equation}\label{eq_maximal_rate_with_latency}
	M(n,\rho, \epsilon) = R(n, \rho - \Delta\rho_m, \epsilon) .
	\end{equation}
\end{lemma}
\begin{proof}
	For fixed rate and blocklength $n$ the maximum allowable decoding time can be calculated using \eqref{eq_total_latency_with_c}. This in turn yields the required power penalty $\Delta\rho$ via \eqref{eq_q_delta_tradeoff} which eventually leads to $ \Delta\rho_m $. Finally, according to \eqref{eq_min_req_power_penalty}, $ M(n,\rho, \epsilon) $ can be determined by shifting the normal approximation by $\Delta\rho_m$ to the right.
\end{proof}

\begin{lemma} \label{lemma_M_monotonicity}
	$ M(n,\rho, \epsilon) $ is monotonically increasing in $ \rho $.
\end{lemma}
\begin{proof}
	Let us introduce the following two maximal rates: $ R(n, \rho^1, \epsilon) $ and $ R(n, \rho^2, \epsilon) $. Suppose that $ \rho^2 \geq \rho^1 $, then using the monotonic structure of the channel capacity \cite{agrell_conditions_for}, $ R(n,\rho^2, \epsilon) \geq R(n,\rho^1, \epsilon) $, and therefore, using Remark 2, $ \Delta\rho_m^2 \geq \Delta\rho_m^1 \geq 0$. Hence, $ M(n,\rho^2, \epsilon) \geq M(n,\rho^1, \epsilon) $.
\end{proof}

\begin{figure}[t]
	\centering
	\includegraphics[width=1\linewidth]{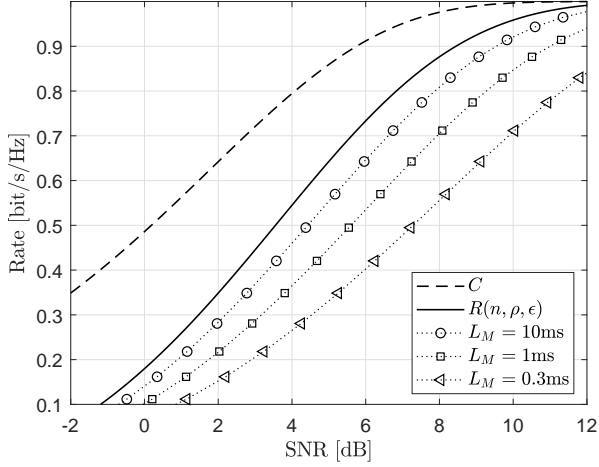}
	\caption{Maximum achievable rates under latency and complexity constraints for $ n=128 $, $ \epsilon = 10^{-5} $, $ T_s = 1\,\mu$s, and $ T_b=1\,$ns.}
	\label{fig_power_vs_rate_w_timeconst}
\end{figure}

In Fig. \ref{fig_power_vs_rate_w_timeconst} the information rate is plotted as a function of the SNR in dB. The dashed line corresponds to the capacity of the BI-AWGN channel and the solid line to the normal approximation for $n=128$ and $\epsilon=10^{-5}$. The remaining three plots in the figure correspond to maximal information rate when latency constraints $L_M=\{10,1,0.3\}\,$ms are imposed. It is assumed that the symbol interval is $T_s=1\,\mu$s and the time required for a binary operation is $T_b=1\,$ns. One can see that the achievability bound shifts to the right more as the constraint on time shrinks. It can be also observed that the gap between normal approximation and $ M(n,\rho, \epsilon) $ widens as $ r $ increases as it is mentioned in Remark 2.

Lemma \ref{lemma_min_power_penalty} and Lemma \ref{lemma_max_inf_rate} reveal that constraints on aggregate latency and  decoding complexity limit the maximal information rate. These results are crucial to understand the capabilities of the communication system and to increase the efficiency. Next, we will discuss some non-trivial optimization problems which affect the efficiency of the communication systems.

\section{Optimal Communication with Latency and Decoding Constraints}
\label{sec_opt_problems}

\subsection{Minimization of Aggregate Latency}

We consider the transmission of a packet that contains a fixed number of information bits, $ k, $ and we are interested in minimizing the aggregate latency, $ L_A $, subject to reliability and transmit power constraints. Such an optimization problem can be encountered in scenarios of industrial control, where, e.g., a sensor transmits a fixed-precision measurement or a control message out of a list of $ 2^k $ possible messages.  The formulation of the problem follows
\begin{subequations}
	\label{eq_opt_problem_1}
	\begin{align}
	\underset{n,\epsilon,\rho_r,\Delta\rho,s}{\text{minimize}} & ~~ L_A
	\\
	\text{s.t.}  & ~~\epsilon \leq \epsilon_m , \label{eq_opt_problem_1_const_1}
	\\
	& ~~ \rho_{r} + \Delta\rho \leq \rho_m, \label{eq_opt_problem_1_const_2}
	\\
	& ~~ k/n \leq R(n,\rho_{r}, \epsilon), \label{eq_opt_problem_1_const_3}
	\\
	& ~~  \rho_r\geq 0, ~ \Delta\rho\geq 0, ~ 0\leq s\leq k, ~ k\leq n. \label{eq_opt_problem_1_const_4}
	\end{align}
\end{subequations}
Here, it is assumed that $ T_b $ and $ T_s $ are fixed. The optimization variables are $ n,\, \epsilon,\, \rho_r,\, \Delta\rho,$ and $ s $. \eqref{eq_opt_problem_1_const_1} and \eqref{eq_opt_problem_1_const_2} represent error rate and power budget constraints, respectively. Lastly, \eqref{eq_opt_problem_1_const_3} indicates the maximal achievable rate without decoding complexity constraints, as given by \eqref{eq_normal_approximation}. 

\begin{lemma} \label{lemma_opt_1_equality}
	The optimum point of \eqref{eq_opt_problem_1} is achieved with equality in (\ref{eq_opt_problem_1_const_1}).
\end{lemma}
\begin{proof}
	We prove the lemma by contradiction. First of all, for fixed $ \rho_r $ and given that $ R(n,\rho_r,\epsilon) \leq R(n,\rho_r,\epsilon_m)$, the feasible set for $ n $ becomes the largest for $ \epsilon=\epsilon_m$. Then assume that the optimal decoder is $\mathrm{d}(\mathcal{C}^*,s^*,\rho_r^*+\Delta\rho^*)$ with $\epsilon(\mathcal{C}^*,s^*,\rho_r^*+\Delta\rho^*) < \epsilon_m$. For some $\sigma>0$ small enough we can find a decoder $\mathrm{d}(\mathcal{C}^*,s^*-\sigma,\rho_r^*+\Delta\rho^*)$ that can achieve $\epsilon(\mathcal{C}^*,s^*-\sigma,\rho_r^*+\Delta\rho^*) = \epsilon_m$. However, the complexity of this decoder is smaller than the optimal one and hence achieves a smaller aggregate latency without violating the CEP constraint.
\end{proof}

The problem now can be further split into a countable sequence of problems, one for every feasible $n$. Fixing $n$ implies that the rate is also fixed, i.e., $r=k/n$. Hence, the reference SNR, $\rho_r$, follows by solving $r=R(n,\rho_r,\epsilon_m)$. It must be noted that a solution to the problem for fixed $n$ can be found only if
\begin{equation}\label{eq_feasibility_condition_opt_1}
\rho_r\leq\rho_m,
\end{equation} 
otherwise the problem is infeasible for the particular $n$. Finally, the problem for fixed $n$, when feasible, can be written as
\begin{subequations}\label{eq_opt_problem_1_sub2}
	\begin{align}
	\underset{\Delta\rho,s}{\text{minimize}} & ~~ K(\mathcal{C},s)
	\\
	\text{s.t.}  & ~~ 0\leq\Delta\rho\leq\rho_m-\rho_r, ~~  0\leq s\leq k .
	\end{align}
\end{subequations}
or equivalently the objective function is the maximization of $ a\sqrt{\Delta\rho}+b $, which is achieved when $\Delta\rho=\rho_m-\rho_r$. The optimal $s$ is given by the following theorem. 

\begin{theorem}
	For a given $ n $, such that the problem is feasible, the corresponding order$ -s $ that minimizes $ l_t $ can be closely approximated to
	\begin{equation}\label{eq_theorem_order_s_selection}
	s \approx \frac{1}{2} \left( k - \sqrt{k^2 + \sqrt[3]{k^2 \eta^4}}\right),
	\end{equation}
	where $ \eta = F\left(\rho_m - \rho_r \right) + 1 - \log n $ .
\end{theorem}
\begin{proof}
	$ F(\Delta\rho) $ is a monotonic decreasing function in $ \Delta\rho $. The complexity of the simplest decoder that meets the constraints can be found while selecting the highest power that is $ \rho_r + \Delta\rho = \rho_m $ and the complexity of this decoder is $ \approx 2^{F\left( \rho_m - \rho_r \right)} $. Finally, \eqref{eq_theorem_order_s_selection} can be obtained by using the same analogy in \eqref{eq_bound_on_s}.
\end{proof}

The optimum selection can be found with exhaustive search over all $ n $ values. A numerical example of the feasible region, denoted as $ S $, for various $ n $ with respect to $ \rho $ is illustrated in Fig. \ref{fig_opt_1}.a where $ k=64 $, $ \rho_m = 5 \,$dB, $\epsilon_m=10^{-5}$, $ T_s = 1 \,\mu$s, and $ T_b = 1 \,$ns. Note that no decoder can be identified until the feasibility condition is met. The optimum, that is shown with a circle, can be found by searching along $ \rho = \rho_m $. A computationally efficient algorithm, linear in $ n $, is proposed in Algorithm 1.

\begin{algorithm}[h]
	\caption{Minimization of $ L_A $} 
	\begin{algorithmic}[1]
		\For {$n=n_{\text{min}},n_{\text{min}+1}, \cdots, n_{\text{max}}$}
		\State\textbf{compute}: $ \rho_r $ from \eqref{eq_reference_power}
		\State\textbf{compute}: $ F(\Delta\rho) $ from \eqref{eq_q_delta_tradeoff}
		\If {$ \rho_m \geq \rho_r $} 
		\State $ K(\mathcal{C},s) = 2^{F(\rho_m - \rho_r)} $
		\Else
		\State $ K(\mathcal{C},s) = \emptyset $
		\EndIf
		\State\textbf{compute}: $ {L_A}(n) = nT_s + kK(\mathcal{C},s)T_b $
		\EndFor
	\end{algorithmic} 
\end{algorithm}

\begin{figure}
	\centering
	\begin{subfigure}[b]{0.49\textwidth}
		\includegraphics[width=1\linewidth]{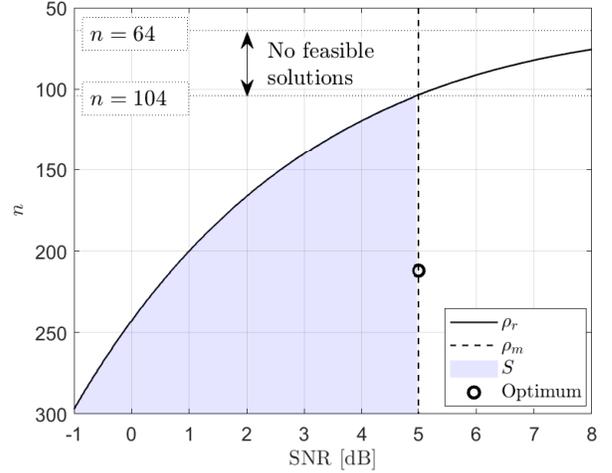}
		\caption{}
		\label{fig_set_realization_opt_1} 
	\end{subfigure}
	
	\begin{subfigure}[b]{0.49\textwidth}
		\includegraphics[width=1\linewidth]{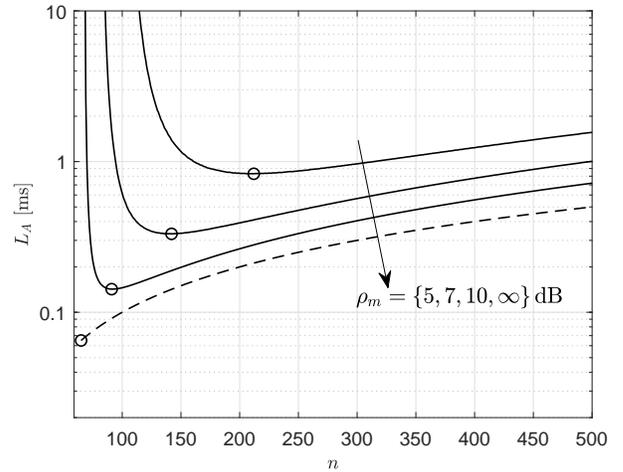}
		\caption{}
		\label{fig_t_t_minimization}
	\end{subfigure}
	
	\caption{(a). Realization of the feasible set $ S $ where $ k=64 $, $ \rho_m = 5 $ dB, $\epsilon_m=10^{-5}$, $ T_s = 1 ~\mu$s, and $ T_b = 1 $ ns. (b). Minimum $ L_A $ with respect to $ n $ for several $ \rho_m $ where $ k=64 $, $\epsilon_m=10^{-5}$, $ T_s = 1 ~\mu$s, and $ T_b = 1 $ ns.}
	\label{fig_opt_1}
\end{figure}

In Fig. \ref{fig_opt_1}.b the aggregate latency is plotted as a function of the codeword length, $n$. It can be seen that for small $n$ the code rate of the selected codebook must be very high. Hence, either the transmission is not possible when the required code rate exceeds \eqref{eq_normal_approximation} or the required decoder must operate close to the normal approximation, which yields high decoding complexity. This translates to very high aggregate latency. As $n$ increases, the required rate is decreasing, hence it is more likely that it can be supported by the power budget or a rate sufficiently far from the normal approximation can be selected. In this case, a decoder with low complexity can be selected and the aggregate latency is dominated by the codeword transmission latency. For power constraints $ \rho_m \in \{5, 7, 10\}\,$dB, the optimal codeword lengths are $ n_{\text{opt}} = \{212, 142, 91\} $, respectively. Infinite $\rho_m$ implies that the symbols are transmitted error free and $ n_{\text{opt}} = k $ since from (\ref{eq_q_delta_tradeoff}), $ kT_b \approx 0 \,$s and hence $ L_A = nT_s $ and linearly increases in $ n $.

\subsection{Minimization of \textit{per-Information-Bit} Energy}

Here, we consider minimizing the \textit{per-information-bit} energy consumption, where the transmission contains a fixed number of information bits, subject to reliability, transmit power, and latency constraints. This optimization problem is significant for communication scenarios where power efficiency is crucial, such as battery powered URLLC systems. A rough analysis may yield the following; minimization of \textit{per-information-bit} energy is proportional to SNR minimization. However, given that a fixed number of $k$ information bits must be transmitted, low SNR values may either lead to theoretically unachievable transmission rates or rates that are very close to the limits and require very complex decoders which may eventually violate the latency constraint. 

The optimization problem can be formulated as
\begin{subequations}
	\label{eq_opt_problem_2}
	\begin{align}
	\underset{n,\epsilon,\rho_r,\Delta\rho,s}{\text{minimize}} & ~~ e_b
	\\
	\text{s.t.}  & ~~\epsilon \leq \epsilon_m , \label{eq_opt_problem_2_const_1}
	\\
	& ~~ L_A \leq L_M \label{eq_opt_problem_2_const_2}
	\\
	& ~~ \rho_{r} + \Delta\rho \leq \rho_m, \label{eq_opt_problem_2_const_3}
	\\
	& ~~ k/n \leq R(n,\rho_{r}, \epsilon), \label{eq_opt_problem_2_const_4}
	\\
	& ~~  \rho_r\geq 0, ~ \Delta\rho\geq 0, ~ 0\leq s\leq k, ~ k\leq n. \label{eq_opt_problem_2_const_5}
	\end{align}
\end{subequations}
where $ e_b = (\rho_{r} + \Delta\rho )/r $ represents the \textit{per-information-bit} energy. Similar to \eqref{eq_opt_problem_1}, it is assumed that the hardware platform is fixed and variables are same. In comparison, an additional aggregate latency constraint is imposed via \eqref{eq_opt_problem_2_const_2}. 

\begin{figure}
	\centering
	\begin{subfigure}[b]{0.49\textwidth}
		\includegraphics[width=1\linewidth]{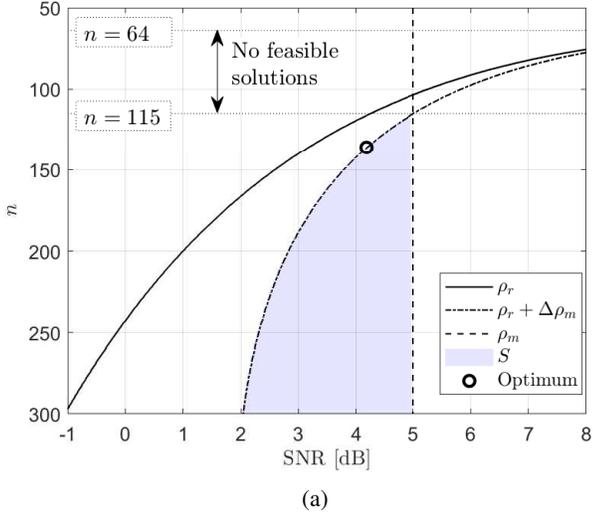}
		\caption{}
		\label{fig_set_realization_opt_2} 
	\end{subfigure}
	
	\begin{subfigure}[b]{0.49\textwidth}
		\includegraphics[width=1\linewidth]{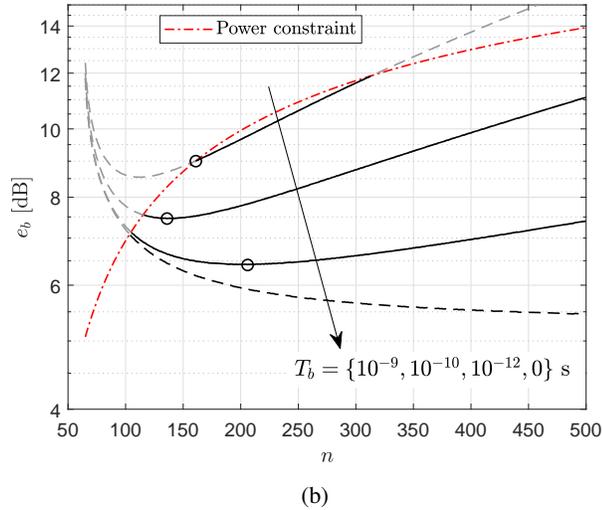}
		\caption{}
		\label{fig_eb_minimization}
	\end{subfigure}
	
	\caption{(a). Realization of the feasible sets for various $ n $ where $ k=64 $, $ \rho_m = 5 $ dB, $\epsilon_m=10^{-5}$, $  L_M = 1$ ms, $ T_s = 1 ~\mu$s, and $ T_b = 0.1 $ ns. (b). Minimum $ e_b $ values for $ k^* = 64 $ for several different complexity constrained receivers where $L_M = 1\,$ms, $ \rho_m=5\, $dB, and $\epsilon_m=10^{-5}$.}
	\label{fig_eb_general}
\end{figure}

\begin{lemma} \label{lemma_opt_2_equality}
	The optimum point of \eqref{eq_opt_problem_2} is achieved with equality in (\ref{eq_opt_problem_2_const_1}).
\end{lemma}
\begin{proof}
	Assume that the optimal decoder is $\mathrm{d}(\mathcal{C}^*,s^*,\rho_r^*+\Delta\rho^*)$ with $\epsilon(\mathcal{C}^*,s^*,\rho_r^*+\Delta\rho^*)<\epsilon_m$. However, for some $ \Delta\rho^* \geq \sigma > 0 $ small enough, one can find a decoder $\mathrm{d}(\mathcal{C}^*,s^*,\rho_r^*+\Delta\rho^*-\sigma)$ that can achieve $\epsilon(\mathcal{C}^*,s^*,\rho_r^*+\Delta\rho^*-\sigma)=\epsilon_m$, which requires lower SNR than the optimal one and hence achieves a smaller \textit{per-information-bit} energy consumption without violating the CEP constraint.
\end{proof}

The power constraint in \eqref{eq_opt_problem_2_const_3} is directly proportional to $ e_b $ and limits it such that $ e_b \leq \rho_m/r  $. Further, we fix $ n $ and split the problem into countable sequence of problems. Now, the rate, $ r $, and the reference SNR, $ \rho_r $, are also fixed. For a feasible $ n $, that meets \eqref{eq_opt_problem_2_const_4} with $ \rho_r\geq 0 $, the problem \eqref{eq_opt_problem_2} now reduces to
\begin{subequations}
	\label{eq_opt_problem_2_sub_1}
	\begin{align}
	\underset{\Delta\rho,s}{\text{minimize}} & ~~ \Delta\rho
	\\
	\text{s.t.}  & ~~ L_A \leq L_M \label{eq_opt_problem_2_sub_1_const_1}
	\\
	& ~~ 0 \leq \Delta\rho \leq \rho_m - \rho_{r}, \label{eq_opt_problem_2_sub_1_const_2}
	\\
	& ~~ 0\leq s\leq k . \label{eq_opt_problem_2_sub_1_const_3}
	\end{align}
\end{subequations}

Without the latency constraint, given in \eqref{eq_opt_problem_2_sub_1_const_1}, \eqref{eq_feasibility_condition_opt_1} gives the feasibility condition. However, selecting $ \Delta\rho $ closer to $ 0 $ corresponds to a decoder with high complexity, which may require longer $ L_D $ for complexity constrained receivers and may violate the latency constraint. 
\begin{lemma} \label{lemma_opt_2_feasibility}
	For a feasible $ n $, there is a set of feasible solutions if  
	\begin{equation}\label{eq_feasibility_condition_opt_2}
	\Delta\rho_m \leq \Delta\rho \leq \rho_m - \rho_r
	\end{equation}
	for $ \Delta\rho \geq 0 $. Thus, the feasibility condition is
	$ \rho_r + \Delta\rho_m \leq \rho_m $ .
\end{lemma}
\begin{proof}
	It is shown in \eqref{eq_min_req_power_penalty} that $ \Delta\rho_m $ gives the minimum amount of power penalty that needs to be paid due to the latency constraint for a fixed CEP. Therefore, selecting the minimum excess power as $ \Delta\rho_m $, guaranties \eqref{eq_feasibility_condition_opt_2}.
\end{proof}

Finally, the optimization problem reduces to
\begin{subequations}
	\label{eq_opt_problem_2_sub_2}
	\begin{align}
	\underset{\Delta\rho,s}{\text{minimize}} & ~~ \Delta\rho
	\\
	\text{s.t.}  & ~~  \Delta\rho_m \leq \Delta\rho \leq \rho_m - \rho_{r}, \label{eq_opt_problem_2_sub_2_const_1}
	\\
	& ~~ 0\leq s\leq k . \label{eq_opt_problem_2_sub_2_const_2}
	\end{align}
\end{subequations} 
Hence, the objective function is minimized when $ \Delta\rho = \Delta\rho_m $. It is worth noting that this operating point lies on $ M(n,\rho,\epsilon_m) $. The corresponding order$ -s $ is given in \eqref{eq_bound_on_s}. An efficient algorithm that solves \eqref{eq_opt_problem_2} is shown in Algorithm 2.

\begin{algorithm}[b]
	\caption{Minimization of $ e_b $} 
	\begin{algorithmic}[1]
		\For {$n=n_{\text{min}},n_{\text{min}+1}, \cdots, n_{\text{max}}$}
		\State\textbf{compute}: $ \rho_r $ from \eqref{eq_reference_power}
		\State\textbf{compute}: $ \Delta\rho_m $ from \eqref{eq_min_req_power_penalty}
		\If {$ \Delta\rho_m + \rho_r \leq \rho_m $} 
		\State $ e{_b}(n) = (\Delta\rho_m +\rho_r) / r  $
		\Else
		\State $ {e_b}(n) = \emptyset $
		\EndIf
		\EndFor
	\end{algorithmic} 
\end{algorithm}

Numerical realizations of the feasible set, $ S $, for various $ n $ are demonstrated in Fig. \ref{fig_eb_general}.a for $ k=64 $, $ \rho_m = 5 $ dB. As seen, no feasible point can be identified unless \eqref{eq_feasibility_condition_opt_2} is satisfied. Notice that the optimum point, depicted with a circle, lies on the $ \rho_r+\Delta\rho_m $ line. 

Minimum $ e_b $ values for different $ T_b $ are depicted in Fig. \ref{fig_eb_general}.b where $L_M = 1\,$ms, $ T_s = 1 ~\mu$s, $ \rho_m=5\, $dB, and $\epsilon_m=10^{-5}$. The red dotted line represents the power constraint and a selection above that line is infeasible. Minimum $ e_b $ values at each $ n $ value are depicted for four different receivers such that $ T_b \in \{0, 0.001, 0.1, 1\} $ ns, where $ T_b = 0 $ ns represents infinite computation power. Notice that, due to the power constraint, for the receiver with $ T_b = 1 $ ns, feasible selections exist only in a small portion of $ n $ and the minimum is located where $ \rho_r+\Delta\rho_m = \rho_m $. For the rest, one can claim that as the hardware capability gets better, i.e. $ T_b $ decreases, the optimum selection of $ n $ increases whereas optimum $ e_b $ decreases.

\subsection{Maximization of Total Transmitted Information Bits}

Next, we investigate the following optimization problem: What is the maximum $ k $ that can be transmitted subject to latency, CEP, and power constraints? This problem is crucial in terms of increasing the efficiency of the communication system and can be formulated as
\begin{subequations}
	\label{eq_opt_problem_3}
	\begin{align}
	\underset{n,k,\epsilon,\rho_r,\Delta\rho,s}{\text{maximize}} & ~~ k
	\\
	\text{s.t.}  & ~~\epsilon \leq \epsilon_m , \label{eq_opt_problem_3_const_1}
	\\
	& ~~ L_A \leq L_M \label{eq_opt_problem_3_const_2}
	\\
	& ~~ \rho_{r} + \Delta\rho \leq \rho_m, \label{eq_opt_problem_3_const_3}
	\\
	& ~~ k/n \leq R(n,\rho_{r}, \epsilon), \label{eq_opt_problem_3_const_4}
	\\
	& ~~  \rho_r\geq 0, ~ \Delta\rho\geq 0, ~ 0\leq s\leq k, ~ k\leq n. \label{eq_opt_problem_3_const_5}
	\end{align}
\end{subequations}

Similar to the previous optimization problems, here we show that optimum solution is achieved with equality in \eqref{eq_opt_problem_3_const_1}. The proof is straightforward by using similar analogy that is shown in Lemma \ref{lemma_opt_1_equality} and Lemma \ref{lemma_opt_2_equality}. 

\begin{figure}
	\centering
	\begin{subfigure}[b]{0.49\textwidth}
		\includegraphics[width=1\linewidth]{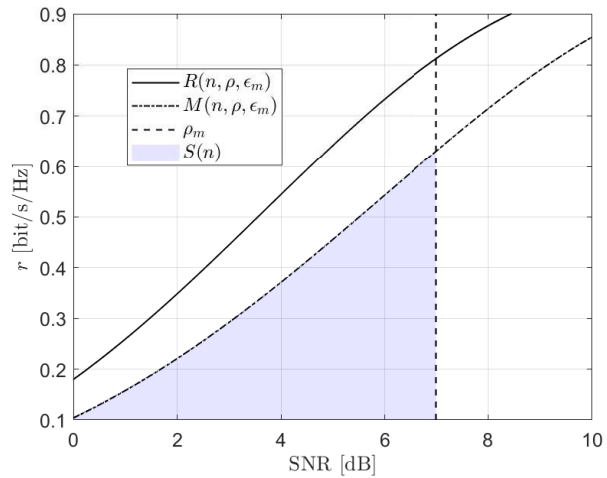}
		\caption{}
		\label{fig_set_realization_opt_3} 
	\end{subfigure}
	
	\begin{subfigure}[b]{0.49\textwidth}
		\includegraphics[width=1\linewidth]{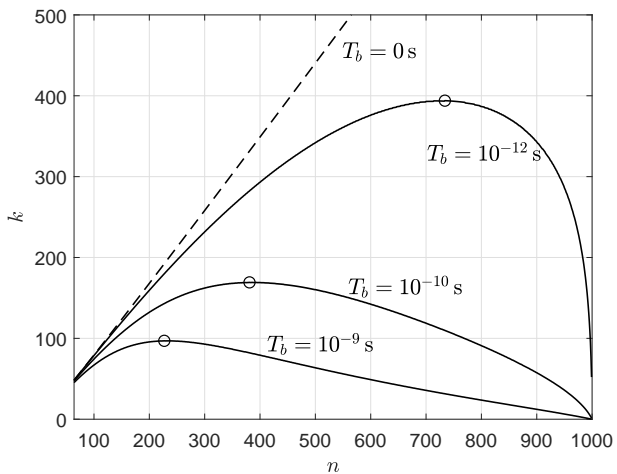}
		\caption{}
		\label{fig_k_maximization}
	\end{subfigure}
	
	\caption{(a). Realization of the feasible set $ S(n=128) $ where $ \rho_m = 7 $ dB, $\epsilon_m=10^{-5}$, $ L_M = 1$ ms, $ T_s = 1 ~\mu$s, and $ T_b = 1 $ ns. (b). Maximum $ k $ for several complexity constrained receivers where $L_M = 1\,$ms, $ \rho_m=7\, $dB, and $\epsilon_m=10^{-5}$. Optimums are shown with circles. }
	\label{fig_opt_3}
\end{figure}

Next, let us first explain the solution to this problem where unlimited computational power is assumed. In this case, a codeword can be decoded instantaneously and therefore $ L_D = 0 $ and all the latency budget can be used for transmission of the codeword, i.e. $n_{\text{inf}}= L_M / T_s$ symbols can be transmitted at a rate that is determined by \eqref{eq_normal_approximation}, which yields
\begin{equation}
k_{\text{inf}} = \big\lfloor n_{\text{inf}} R\left(n_{\text{inf}},\rho_m,\epsilon_m\right) \big\rfloor .
\end{equation}
Notice that SNR is chosen to be $ \rho_m $ due to the monotonic structure of the channel \cite{agrell_conditions_for}. 

However, with decoding complexity constraints, the following trade-off arises. If $n$ is selected small, the available duration for decoding can be sufficient so that a high rate code can be used. As $n$ increases, the available duration for decoding shrinks and a code with decreasing code rate must be selected so that the aggregate latency constraint is satisfied. The solution of such a problem for complexity constrained receivers is not trivial and may need a comprehensive search with various parameters.

Without loss of generality, let us first set $ \epsilon = \epsilon_m $ and fix $ n $ and split the problem into a countable sequence of subproblems. It should be noted that $L_M/T_s$ is an upper bound of $n$. It is shown in Lemma \ref{lemma_opt_2_feasibility} that the latency constraint in \eqref{eq_opt_problem_3_const_2} can be converted to a power penalty constraint. There, it is also shown that the feasibility constraint is $ \rho_r + \Delta\rho_m  \leq \rho_m$. Here, we further extend and instead of converting the latency constraint into a power constraint, using Lemma \ref{lemma_min_power_penalty}, we convert it to a rate constraint. Thus, the problem reduces to
\begin{subequations}
	\label{eq_opt_problem_3_sub_2}
	\begin{align}
	\underset{k,\rho_r,\Delta\rho,s}{\text{maximize}} & ~~ k
	\\
	\text{s.t.}  & ~~ \Delta\rho \leq \rho_m - \rho_{r}, \label{eq_opt_problem_3_sub_2_const_1}
	\\
	& ~~ k/n \leq M(n,\rho_{r}, \epsilon_m), \label{eq_opt_problem_3_sub_2_const_2}
	\\
	& ~~  \rho_r\geq 0, ~ \Delta\rho\geq 0, ~ 0\leq s\leq k, ~ k\leq n. \label{eq_opt_problem_3_sub_2_const_3}
	\end{align}
\end{subequations}

Numerical realization of such a problem is demonstrated in Fig. \ref{fig_opt_3}.a where $ n $ is fixed to $ 128 $ and $ \epsilon_m = 10^{-5} $, $ \rho_m = 7 \,$dB, $ L_M = 1 \,$ms, $ T_s = 1\,\mu $s, and $ T_b = 1 \,$ns. The feasible set is shown with $ S(n) $. Notice that, due to Lemma \ref{lemma_M_monotonicity}, the sub-optimum rate-power selection is the topmost point of the set $ S(n) $, which is also the junction point of $ M(n,\rho, \epsilon_m) $ and $ \rho = \rho_m $, that is $ M(n,\rho_m, \epsilon_m) $. Hence, the solution to the optimization problem in (\ref{eq_opt_problem_3}) follows 
\begin{equation}\label{key}
k^{\text{opt}} = \big\lfloor n^\text{opt} M(n^\text{opt},\rho_m, \epsilon_m) \big\rfloor.
\end{equation}
where $ n^\text{opt} $, the optimum $ n $ that maximizes $ k $, follows $ n^{\text{opt}} = \underset{\{n| n \in \mathbb{N^+}\}}{\mathrm{arg~max}} ~~ n M(n,\rho_m, \epsilon_m) $. 
A computationally efficient algorithm, linear in $ n $, is proposed in Algorithm 3.

\begin{algorithm}[h]
	\caption{Maximization of $ k $} 
	\begin{algorithmic}[1]
		\For {$n=n_{\text{min}},n_{\text{min}+1}, \cdots, n_{\text{max}}$}
		\State\textbf{compute}: $ R(n,\rho,\epsilon_m) $ using \eqref{eq_normal_approximation}
		\State\textbf{compute}: $ \Delta\rho_m $ using \eqref{eq_maximal_rate_with_latency}, $ \forall r \in (0,1] $
		\State\textbf{compute}: $ M(n,\rho,\epsilon_m) $ using \eqref{eq_maximal_rate_with_latency}
		\State\textbf{compute}: $ k(n) = \big\lfloor n M(n,\rho_m,\epsilon_m) \big\rfloor $
		\EndFor
	\end{algorithmic} 
\end{algorithm}

In Fig. \ref{fig_opt_3}.b numerical results that correspond to the investigated scenario are plotted for  $ L_M = 1\, $ms, $ \rho_m=7\,$dB, and $ \epsilon_m=10^{-5} $. Four different choices for execution times for a binary operation are shown: $ T_b \in \{0, 0.001, 0.1, 1\} \,$ns. The previously introduced trade-off is clear here and the maximums appear at $ n^{\text{opt}} = \{ 227, 381, 734, 1000\} $, respectively. Corresponding $ k^{\text{opt}} $ values are $ k^{\text{opt}} = \{96, 169, 393, 901\}$. Ratios of $ k^{\text{opt}} $ values found for complexity constrained receivers to the $ k^{\text{opt}} $ of infinite computation power receiver are $ \approx 0.1, 0.18, 0.43 $, respectively. Thus, one can conclude that if complexity constraints and decoding duration are taken into account, depending on the receiver capabilities, the maximum achievable values are much less than the theoretical limits.

\section{Discussion}
\label{sec_discussions}

\subsection{Other Families of Codes}

Thus far we have reviewed and solved several optimization problems for URLLC applications with complexity constrained OS decoders. Solutions to these problems depend on the model that is introduced in Section \ref{sec_complexity_constrained} where the trade-off between computational complexity and power penalty for a fixed reliability constraint is modeled in a simple way. Although \eqref{eq_q_delta_tradeoff} was derived based on linear block encoder and OS decoders, results in the literature, \cite[Fig. 6]{shivarnimoghaddam_short_block}, \cite[Fig. 6.1 to Fig. 6.9]{lian_performance}, reveal that when it comes to the relation between computational complexity and power penalty in the short block-length regime, other families of codes follow a similar pattern. Here, we further extend our conclusions to the following coding schemes, which are considered as promising solutions for URLLC applications: \textit{i}) TBCCs with list Viterbi decoding \cite{wang_list_viterbi}, \textit{ii}) polar codes under 7-bit cyclic-redundancy-check aided successive cancellation list (SCL) decoding \cite{liva_code_design,niu_crc_aided}, \textit{iii}) binary LDPC codes with min-sum decoder \cite{zarkeshvari_on_implemantation}, and illustrate that their behaviour can be closely modelled by \eqref{eq_q_delta_tradeoff}.\footnote{The field of practical codes is extensive and many tricks can be used to reduce the complexity of a decoder. However, this is beyond the scope of this study. Here, there is no intent to find the optimal decoder in terms of  computational complexity. Instead, we consider reasonably optimized off-the-shelf codes which are also considered as promising solutions for URLLC applications and illustrate that their computational complexity versus power gap behavior follow similar trends to the OS decoder.}

Recall that the complexity-reliability trade-off in OS decoders is controlled with the order$ -s $, whereas this parametrization in TBCCs, polar codes, and LDPC codes can be controlled by the memory size, $ \texttt{M} $, list size, $ \texttt{L} $, and maximum iteration number, $ \texttt{I} $, respectively. Therefore, performance of TBCC with $ \texttt{M} = \{1, 2, 4, 6, 8, 10 ,12, 14\} $ \cite{hehn_ldpc_codes, liva_code_design}, SCL for polar codes with $ \texttt{L} = \{1, 2, 4, 8, 16\} $, and finally min-sum decoding for LDPC with $ \texttt{I} = \{1, 2, 5, 10, 20, 50, 100, 250\} $ are investigated.\footnote{It is worth to note that although the parameter change at the decoder does not effect the encoder structure and complexity in linear coding schemes, it may change the convolutional encoder and increase or decrease its complexity.} CEP results of these codes for $ n=128 $, $ k=64 $ are not shown due to page limitations, however it is seen that performance of TBCC codes is approaching the normal approximation as $ \texttt{M} $ increases. Performance of LDPC and polar codes improves as $ \texttt{I} $ and $ \texttt{L} $ increase. However, although polar codes are performing better than LDPC codes, in terms of CEP, TBCC outperforms both of them, of course, at the expense of a high decoding complexity, which are shown in the next figure. 

\begin{figure}[t]
	\centering
	\includegraphics[width=1\linewidth]{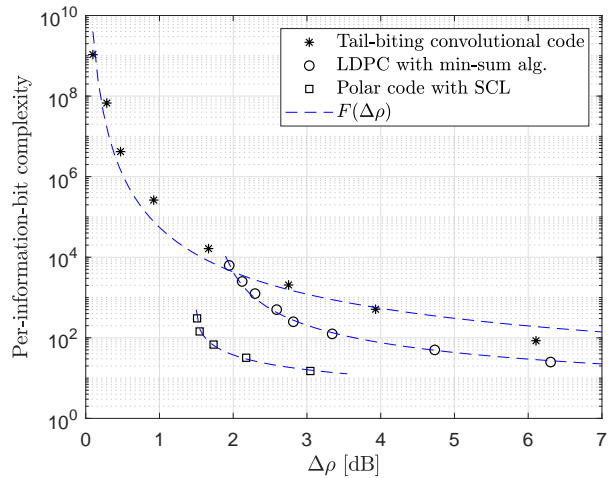}
	\caption{Power penalty values of TBCC, LDPC, and polar codes versus their complexities for $ n=128 $, $ k=64 $ where $ \epsilon=10^{-5} $.}
	\label{fig_q_vs_deltaP_all_codes}
\end{figure}

In Fig. \ref{fig_q_vs_deltaP_all_codes} the total number of binary operations \textit{per-information-bit} is plotted as a function of the power penalty for TBCC, LDPC, and polar codes to achieve $ \epsilon = 10^{-5} $, where the \textit{per-information-bit} complexities of the decoders are obtained from \cite{sybis_channel_coding}. Although, values of $ \texttt{M} $, $ \texttt{L} $, and $ \texttt{I} $ are not depicted, as one can predict, these values are increasing as the decoder approaches the bound. One can see that the trade-off between complexity vs. power penalty for TBCC, LDPC, and polar codes can be closely pursued with the proposed model given in \eqref{eq_q_delta_tradeoff}. Hence, it can be advocated that \eqref{eq_q_delta_tradeoff} is a useful proxy for a general study of URLLC systems with computational complexity constraints.

\subsection{Parallel Processing}

An important feature that can significantly reduce the decoding duration is the availability of parallel processing, which is the possibility of executing multiple computational processes in multiple processors simultaneously. Implementation of parallel processing depends on the parallelizability of the computational task. Suppose that a fraction $ \alpha $ of a computational task is parallelizable, meaning that only the fraction $ \alpha $ of the total task before parallelization can be executed in parallel, whereas the fraction $ (1-\alpha) $ of the task needs to be run sequentially. The theoretical upper bound on the speed-up of the execution duration is addressed by the Amdahl's law \cite{amdahl_validity, hill_amdahls}
\begin{equation}\label{eq_amdahls_law}
U = \frac{L_D}{L_D^{P}} = \frac{1}{ \frac{\alpha}{P} + (1-\alpha) } ,
\end{equation}
where $ U $, $ L_D $, $L_D^{P} $, and $ P $ represent the speed-up in time of the computational task, the total decoding duration on a single processor, the total decoding duration with parallel processing, and the number of parallel processors, respectively. Thus, optimally, the execution time of a task with parallel processing is $ U $ times faster than running the same task sequentially. 

Using \eqref{eq_total_latency_with_c} and \eqref{eq_amdahls_law}, one can relate the speed-up coefficient $ U $ to $ T_b $ as the following
\begin{equation}\label{eq_tb_parallel}
T_b^{P} = \frac{1}{U} T_b ,
\end{equation}
where $ T_b^{P} $ can be named as the average time required for a binary operation in parallel processing, in which all processors in parallel are identical and the execution time of a binary process is $ T_b $ for all. Hence, \eqref{eq_tb_parallel} shows that the effect of parallel computation is linear in $ T_b $ and therefore does not change the analysis in Section \ref{sec_opt_problems} but may change the numerical results since although $ T_b $ is only related with the hardware platform, $ T_b^{P} $ depends on the fraction $ \alpha $.

\section{Conclusions}

The aggregate latency caused by codeword transmission and decoding is the main focus in this study. The empirical model we have presented in this paper can accurately show the trade-off between complexity of OS decoders versus their power gap to the normal approximation. Based on the insights from the proposed model, maximal achievable transmission rates with OS decoders under stringent latency and computational complexity constraints are presented. In particular, our results highlight the effects of these constraints on transmission parameters and hence show that decoding time has a considerable effect on the bounds of the short block-length codes when complexity constraints are taken into account. Next, several optimization problems that aims to increase the transmission efficiency of the URLLC system with OS decoder have been formulated and solved. It is shown that when complexity constraint and decoding duration are considered in a low-latency communication scenario, the optimum selections of the transmission parameters vary significantly compared to the unconstrained decoder scenarios.

\balance

\bibliographystyle{IEEEtran}
\bibliography{bare_jrnl_comsoc}

\end{document}